\documentclass[a4paper,10pt]{article}

\usepackage{graphicx}
\usepackage{geometry}
\usepackage{amsthm}
\usepackage{amssymb}
\usepackage{amsmath}
\usepackage{hyperref}
\usepackage{wrapfig}
\usepackage{tipa}
\usepackage{tikz}
\usetikzlibrary{arrows.meta}
\usepackage{mycommands}
\usepackage{enumitem}

\usepackage[font=small,labelfont=bf]{caption}
\usepackage[font=small,labelfont=normalfont,labelformat=simple]{subcaption}

\usepackage{bookmark}

\graphicspath{{./fig/}}

\newtheorem{theorem}{Theorem}
\newtheorem{definition}[theorem]{Definition}
\newtheorem{lemma}[theorem]{Lemma}
\newtheorem{corollary}[theorem]{Corollary}

\author
{
Torsten M\"utze\thanks{Department of Computer Science, University of Warwick, United Kingdom, \texttt{torsten.mutze@warwick.ac.uk}.
Torsten M\"utze is also affiliated with Charles University, Faculty of Mathematics and Physics, and was supported by GACR grant GA~19-08554S, and by DFG grant~413902284.}
\and
Manfred Scheucher\thanks{Institut f\"ur Mathematik, Technische Universit\"at Berlin, Germany, \texttt{scheucher@math.tu-berlin.de}}\hspace{1.5mm}
}

\date{}

\title{On L-shaped point set embeddings of trees: \\
first non-embeddable examples\footnote{An extended abstract of this paper appeared in the Proceedings of the 26th International Symposium on Graph Drawing and Network Visualization~\cite{ScheucherMuetze18}.}}

\begin{document}
\maketitle

\begin{abstract}
An \emph{L-shaped} embedding of a tree in a point set is a planar drawing of the tree where the vertices are mapped to distinct points and every edge is drawn as a sequence of two axis-aligned line segments.
There has been considerable work on establishing upper bounds on the minimum cardinality of a point set to guarantee that any tree of the same size with maximum degree~4 admits an L-shaped embedding on the point set.
However, no non-trivial lower bound is known to this date, i.e., no known $n$-vertex tree requires more than $n$~points to be embedded.
In this paper, we present the first examples of $n$-vertex trees for $n\in\{13,14,\allowbreak{}16,\allowbreak{}17,\allowbreak{}18,\allowbreak{}19,20\}$ that require strictly more points than vertices to admit an L-shaped embedding.
Moreover, using computer help, we show that every tree on~$n\leq 12$ vertices admits an L-shaped embedding in every set of $n$~points.
We also consider embedding \emph{ordered} trees, where the cyclic order of the neighbors of each vertex in the embedding is prescribed.
For this setting, we determine the smallest non-embeddable ordered tree on $n=10$ vertices, and we show that every ordered tree on $n\leq 9$ or $n=11$ vertices admits an L-shaped embedding in every set of $n$~points.
We also construct an infinite family of ordered trees which do not always admit an L-shaped embedding, answering a question raised by Biedl, Chan, Derka, Jain, and Lubiw.
\end{abstract}

\section{Introduction}
\label{sec:intro}

An \emph{L-shaped} embedding of a tree in a point set is a planar drawing of the tree where the vertices are mapped to distinct points of the set and every edge is drawn as a sequence of two axis-aligned line segments; see Figure~\ref{fig:example}.
Here and throughout this paper, all point sets are such that no two points have the same $x$- or $y$-coordinate.
The investigation of L-shaped embeddings was initiated in~\cite{KanoSuzuki2011,FinkEtAl2012,GiacomoFFGK13}.
In particular, Di~Giacomo et al.~\cite{GiacomoFFGK13} showed that $O(n^2)$ points are always sufficient to embed any $n$-vertex tree.
Note that an L-shaped embedding requires that the maximum degree of the tree is at most~4.
Moreover, if the maximum degree is~2, then the tree is a path and can be embedded greedily on any point set of the same size.
Formally, let~$f_d(n)$ denote the minimum number~$N$ of points such that every $n$-vertex tree with maximum degree~$d\in\{3,4\}$ admits an L-shaped embedding in every point set of size~$N$.

\begin{figure}[htb]
\centering
\includegraphics{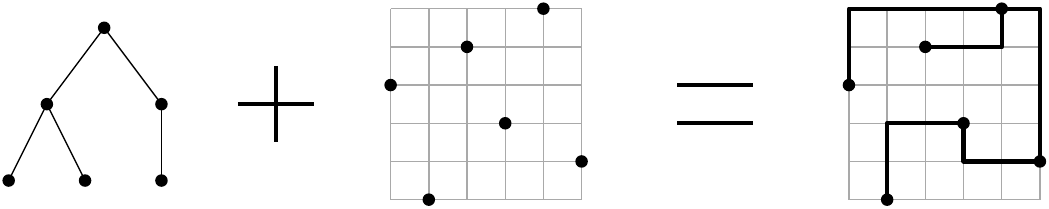}  
\caption{An L-shaped embedding of a tree in a point set.}
\label{fig:example}
\end{figure}

The second author's master's thesis~\cite{scheucher2015} proposed a method to recursively construct an L-shaped embedding of any $n$-vertex tree in any point set of size~$O(n^{1.58})$ (see also~\cite{ahs16}).
Biedl et al.~\cite{BiedlCDJL17} improved on this, proving that $f_3(n)=O(n^{1.22})$ and $f_4(n)=O(n^{1.55})$ points are enough.
To this date, no lower bound besides the trivial bound~$f_d(n) \geq n$ is known, i.e., no known $n$-vertex tree requires more than $n$~points to be embedded.
Di Giacomo, Frati, Fulek, Grilli, and Krug~\cite{GiacomoFFGK13} specifically asked for a tree and point set that would prove $f_4(n)>n$.
The same question was reiterated by Fink, Haunert, Mchedlidze, Spoerhase, and Wolff~\cite{FinkEtAl2012}, and by Biedl, Chan, Derka, Jain, and Lubiw~\cite{BiedlCDJL17}.
Kano and Suzuki~\cite{KanoSuzuki2011} even conjectured that $f_3(n)=n$.

Biedl et al.~\cite{BiedlCDJL17} also considered a more restricted setting of embedding \emph{ordered} trees, where the cyclic order of the neighbors of each vertex in the embedding is prescribed.
They presented a 14-vertex ordered tree which does not admit an L-shaped embedding in a point set of size~14\footnote{
Specifically, their counterexample is the 14-vertex caterpillar with 6~vertices on the central path and a pending edge on each side of the four inner vertices of the path.
The point set is a $(4,6,4)$-staircase in our terminology (see Definition~\ref{def:staircase}).
}, and they raised the problem to find an infinite family of such non-embeddable ordered trees.

\subsection{Our results}

We begin presenting our results for the setting where there are no constraints on the cyclic order in which the neighbors appear around each vertex of the tree.
With brute-force computer search, we verified that all trees on $n\leq 12$ vertices can be embedded in every point set of size~$n$.

\begin{theorem}[Computer-based]
\label{thm:small-trees}
Every tree on~$n \le 12$ vertices admits an L-shaped embedding in every set of $n$~points.
\end{theorem}

We also formulated a SAT instance to test a given pair of tree and point set for embeddability.
This way, we found a 13-vertex tree that does not admit an embedding in a particular point set.

\begin{theorem}
\label{thm:T13}
The tree~$T_{13}$ in Figure~\ref{fig:T13} does not admit an L-shaped embedding in the point set~$S_{13}$ shown in the figure.
\end{theorem}

Even though the 13-vertex tree~$T_{13}$ was found using the help of a SAT solver, a human-verifiable proof of Theorem~\ref{thm:T13} is not hard to obtain.

Besides the pair~$(T_{13},S_{13})$, we also found pairs of trees and point sets that do not admit an embedding for larger values of~$n$.
Overall, we found pairs of $n$-vertex trees and point sets of size~$n$ for $n\in\{13,14,16,17,18,19,20\}$.
For~$n=15$, however, our computer search did not yield any non-embeddable example (the search was not exhaustive).
We remark that all known non-embeddable trees contain~$T_{13}$ as a subtree.

We now focus on the more restricted setting of ordered trees introduced in~\cite{BiedlCDJL17}, where the cyclic order of the neighbors of each vertex in the embedding is prescribed.

\begin{figure}[t]
\centering
\includegraphics{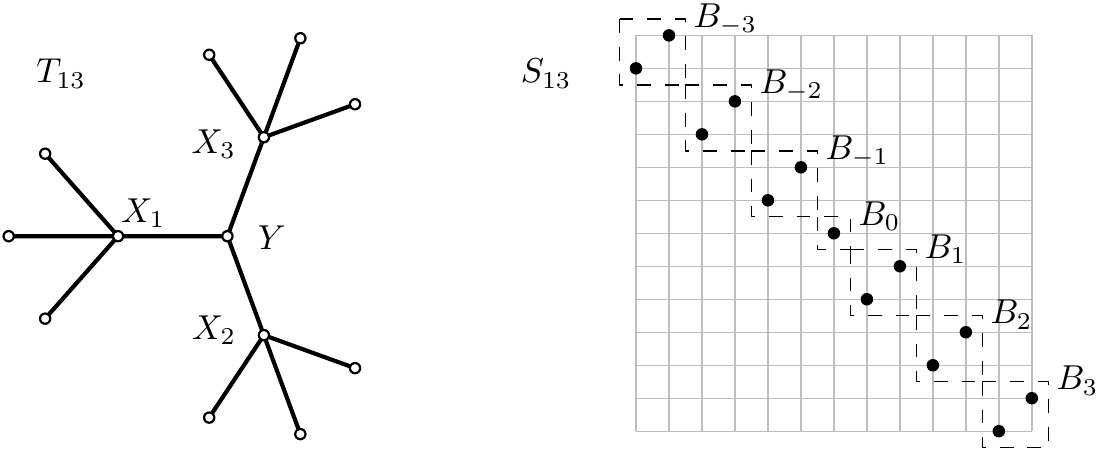}
\caption{The tree~$T_{13}$ (left) does not admit an L-shaped embedding in the $(2,2,2,1,2,2,2)$-staircase point set~$S_{13}$ (right).
The boxes $B_{-3},\ldots,B_3$ are highlighted by dashed frames.
}
\label{fig:T13}
\end{figure}

\begin{figure}[t]
\centering
\includegraphics{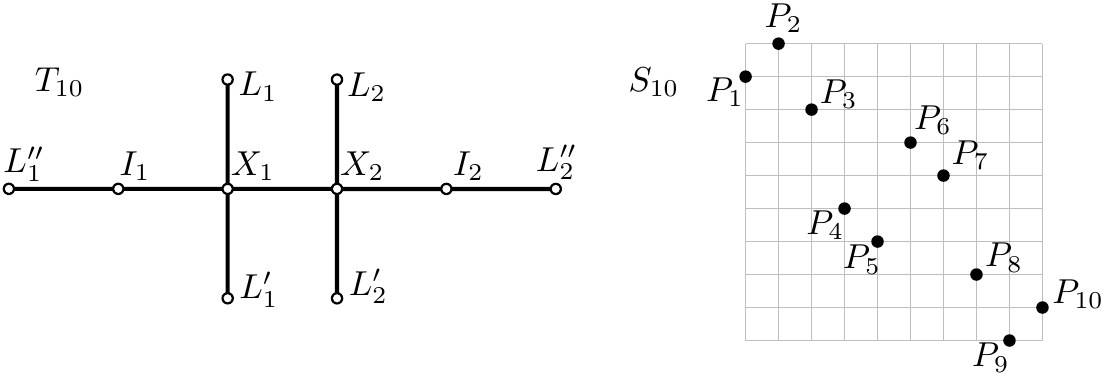}
\caption{The ordered tree~$T_{10}$ (left) does not admit an L-shaped embedding in the point set~$S_{10}$ (right).}
\label{fig:T10}
\end{figure}

\begin{theorem}[Computer-based]
\label{thm:small-ordered-trees}
Every ordered tree on~$n \leq 9$ vertices or on $n=11$ vertices admits an L-shaped embedding in every set of $n$~points.
\end{theorem}

We also found a 10-vertex tree that does not admit an embedding in a particular point set.
This is a smaller non-embeddable instance than the one for~$n=14$ previously presented in \cite{BiedlCDJL17}.

\begin{theorem}
\label{thm:T10}
The ordered tree~$T_{10}$ in Figure~\ref{fig:T10} does not admit an L-shaped embedding in the point set~$S_{10}$ shown in the figure.
\end{theorem}

Remarkably, the pair $(T_{10},S_{10})$ is the only one on~$n=10$ vertices/points not admitting an L-shaped embedding.

Moreover, we construct an infinite family of ordered trees that do not admit an L-shaped embedding on certain point sets, answering a question raised by Biedl, Chan, Derka, Jain, and Lubiw in~\cite{BiedlCDJL17}.
As it turns out, the point sets that appear to be difficult for embedding have a regular staircase shape as shown in Figure~\ref{fig:Tr} (see also Figure~\ref{fig:T13}).

\begin{definition}[Staircase point set] 
\label{def:staircase}
For any integer $n$ with $n = a_1 + \cdots + a_k$ where $a_1,\ldots,a_k \in \mathbb{N}$, the \emph{$(a_1,\ldots,a_k)$-staircase} is the point set consisting of a sequence of $k$~disjoint boxes, ordered from top-left to bottom-right, and the $i$th box contains a sequence of $a_i$~points with increasing $x$- and $y$-coordinate.
\end{definition}

\begin{theorem}
\label{thm:Tr}
For any even~$r\geq 10$, the ordered tree~$T_r^*$ on $n = 9r+8$ vertices in Figure~\ref{fig:Tr} does not admit an L-shaped embedding in the $n$-point $(2,\ldots,2)$-staircase.
\end{theorem}

We conjecture that~$T_r^*$ does not admit an embedding in the same point set, even when considered as an unordered tree, i.e., in the original unrestricted setting.

\begin{figure}[t]
\centering
\includegraphics{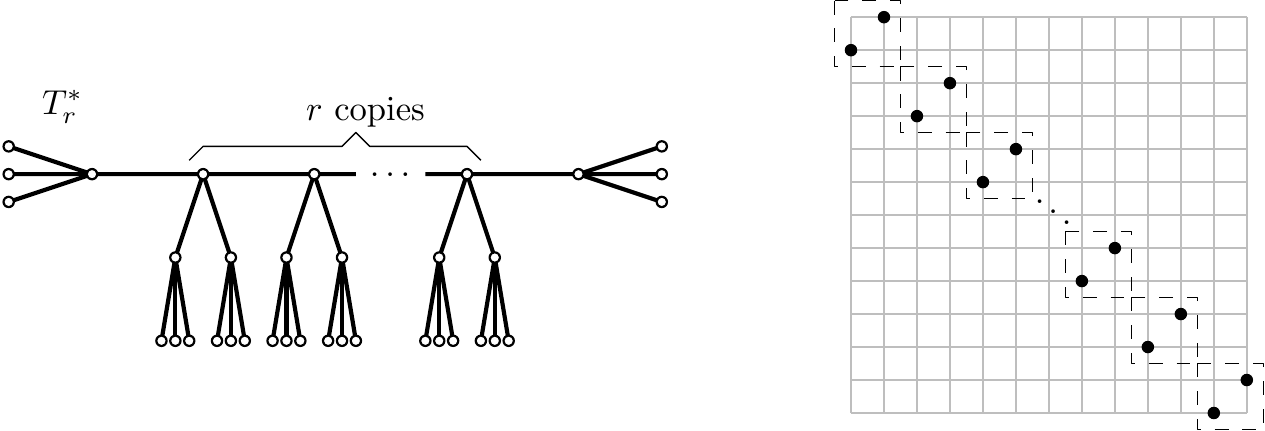}
\caption{A family of ordered trees~$T_r^*$ (left) that does not admit an L-shaped embedding in the $n$-point $(2,\ldots,2)$-staircase (right) for $n=9r+8$ and even $r\geq 10$.
The boxes of the point set are highlighted.}
\label{fig:Tr}
\end{figure}

\subsection{Related work}

Besides the problem of finding L-shaped embeddings of arbitrary trees in arbitrary point sets, various special classes of trees and point sets have also been studied.
For instance, perfect binary and perfect ternary $n$-vertex trees can be embedded in any point set of size~$O(n^{1.142})$ or~$O(n^{1.465})$, respectively~\cite{BiedlCDJL17}.
Moreover, trees with pathwidth~$k$ can be embedded in any set of $2^kn$~points~\cite[Chapter~3.3.2]{scheucher2015} (see also~\cite{ahs16}).\footnote{
For the definition of pathwidth, we refer the reader to~\cite{RobertsonSeymour83}.}
Further, any $n$-vertex caterpillar with maximum degree~3 can be embedded in any point set of size~$n$~\cite{GiacomoFFGK13}.
A \emph{caterpillar} is a tree with the property that all leaves are in distance~1 of a central path.
For maximum degree~4 caterpillars, the currently best known upper bound is $4n/3+O(1)$ many points~\cite[Chapter~5.2.1]{scheucher2015}.
Biedl et al.~\cite{BiedlCDJL17} showed that any ordered caterpillar can be embedded in any point set of size~$O(n\log n)$.

When point sets are chosen uniformly at random, i.e., the $y$-coordinates are a random permutation, it is known that $O(n \log n (\log \log n)^2)$ and $O(n^{1.332})$ points are sufficient to embed any tree with maximum degree~3 or~4, respectively, with probability at least~$1/2$~\cite[Chapter~4]{scheucher2015} (see also~\cite{ahs16}).

Another known setting are non-planar L-shaped point set embeddings, where L-shaped edges are allowed to cross properly, but edge-segments must not overlap.
For this setting, it is known that $n$~points are sufficient to embed any $n$-vertex tree with maximum degree~3~\cite{FinkEtAl2012,GiacomoFFGK13} or any $n$-vertex caterpillar with maximum degree~4~\cite[Theorem~21]{scheucher2015}.
For $n$-vertex trees with maximum degree~4 the currently best upper bound on the required number of points is $7n/3+O(1)$~\cite[Theorem~7]{scheucher2015}.

\subsection{Outline of this paper}

In Section~\ref{sec:key} we present a key lemma that is used repeatedly in our constructions.
In Sections~\ref{sec:T13proof} and \ref{sec:T10proof} we present the proofs of Theorems~\ref{thm:T13} and \ref{thm:T10}, respectively.
Section~\ref{sec:Trproof} is devoted to proving Theorem~\ref{thm:Tr}.
We describe our computational approach to proving Theorems~\ref{thm:small-trees} and \ref{thm:small-ordered-trees} by exhaustive search in Section~\ref{sec:comp}.
More non-embeddable small trees are presented in Section~\ref{sec:examples}, together with our SAT model which is used to verify non-embeddability.
We conclude in Section~\ref{sec:prob} with some challenging open problems.

\section{Key lemma}
\label{sec:key}

The following key lemma is used repeatedly in our arguments about non-embeddability of unordered trees.
It asserts that in an L-shaped embedding, two tree vertices of degree~4 cannot both be mapped to the two points in a box of size~2 in a staircase point set.
The size here refers to the number of points in the box, not to the width or height.
Our examples in Theorem~\ref{thm:T13}, Theorem~\ref{thm:Tr}, and the ones in Section~\ref{sec:examples} are all constructed by considering trees with many degree~4 vertices, and staircase point sets with many boxes of size~2, which creates many constraints.

\begin{lemma}
\label{lem:box}
Let $T$ be an unordered tree with two vertices~$X_1$ and~$X_2$ of degree~4.
Not both~$X_1$ and~$X_2$ can be mapped to the two points in a box of size~2 in a staircase point set.
\end{lemma}

\begin{proof}
For the sake of contradiction assume that both~$X_1$ and~$X_2$ are mapped to the two points in a box of size~2 in the staircase point set.
W.l.o.g.\ we may assume that~$X_2$ is above and to the right of~$X_1$.
As $X_1$ and~$X_2$ both have degree~4, each of them has edges incident to its left, right, bottom and top.
Consider the two edges incident to the top and the right of~$X_1$, and the two edges incident to the bottom and left of~$X_2$.
Among these edges, at most one can be an edge connecting~$X_1$ and~$X_2$ (provided they are adjacent in~$T$).
Consequently, one of these edges connects~$X_1$ to another vertex outside this box, and one of these edges connects~$X_2$ to another vertex outside this box.
As all points outside this box are either above and to the left of it or below and to the right of it, these two edges must cross, a contradiction; see Figure~\ref{fig:box}.
\end{proof}

\begin{figure}[h]
\centering 
\includegraphics[page=1]{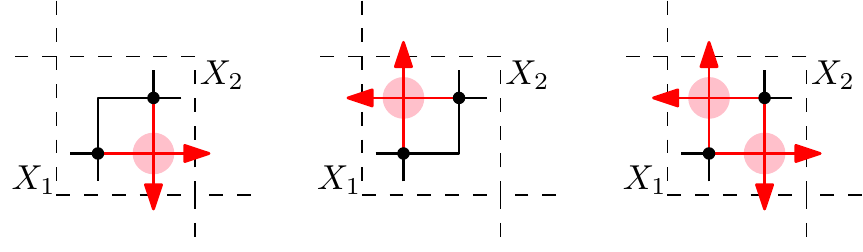}
\caption{
Illustration of the proof of Lemma~\ref{lem:box}.
Crossing edges are highlighted.
}
\label{fig:box}
\end{figure}

\section{Proof of Theorem~\ref{thm:T13}}
\label{sec:T13proof}

Consider the (unordered) tree~$T_{13}$ and the $(2,2,2,1,2,2,2)$-staircase point set~$S_{13}$ depicted in Figure~\ref{fig:T13}.
We label the degree-3 vertex of~$T_{13}$ by~$Y$ and the three degree-4 vertices of~$T_{13}$ by~$X_1,X_2,X_3$.
Moreover, we label the boxes in the staircase point set~$S_{13}$ from left to right by~$B_{-3},B_{-2},\ldots,B_{3}$.
Note the symmetry of~$T_{13}$, as the vertex~$Y$ joins three isomorphic subtrees.
Moreover, $S_{13}$ has reflection symmetries along both diagonals of the grid.

For the sake of contradiction, we assume that an L-shaped embedding of~$T_{13}$ in~$S_{13}$ exists.
We first derive three lemmas that capture which boxes the vertices~$X_1,X_2,X_3,Y$ can be mapped to in such an embedding, and we then complete the proof by distinguishing two main cases.

In the embedding, the L-shaped edge between any two neighboring vertices of the tree can have one of four possible orientations, and we refer to it as an \necorner{}-, \secorner{}-, \swcorner{}-, or \nwcorner{}-edge.

\begin{lemma}
\label{lem:X-not-on-boundary}
Neither of the four vertices~$X_1,X_2,X_3,Y$ is mapped to~$B_{-3}$ or to~$B_{3}$.
\end{lemma}

\begin{proof}
All points in~$B_{-3}$ and~$B_{3}$ lie on the bounding box of the point set, so if one of the~$X_i$ is mapped to such a point, then one of the four edges incident with~$X_i$ would leave the bounding box, which is impossible.
Moreover, $Y$ cannot be mapped to one of these two boxes, as otherwise one of the~$X_i$, which are the only neighbors of~$Y$ in~$T_{13}$, would be mapped to the other point of that same box.
\end{proof}

Lemma~\ref{lem:box} immediately gives the following result.

\begin{lemma}
\label{lem:X-distinct-boxes}
Each of the degree-4 vertices~$X_i$ is mapped to a distinct box.
\end{lemma}

\begin{lemma}
\label{lem:Y-inside-Xs}
Not all three points~$X_1,X_2,X_3$ lie on the same side (above, below, left, or right) of~$Y$.
\end{lemma}

\begin{proof}
It suffices to prove one of the statements, then the others follow by symmetry.
Suppose for the sake of contradiction that~$X_1,X_2,X_3$ all lie above~$Y$.
As one edge is incident to the right of~$Y$, one of the~$X_i$, say~$X_3$, is mapped to the same box, and~$Y$ is below and to the left of~$X_3$ in that box; see Figure~\ref{fig:T13proof}.
Moreover, $YX_3$ is an \secorner{}-edge.
As~$X_3$ has degree~4, and each box contains at most two points, the edge incident to the top of~$Y$ that connects~$Y$ to~$X_1$ or~$X_2$ crosses the edge incident to the left of~$X_3$, a contradiction.
\end{proof}

\begin{figure}[h]
\centering
\includegraphics[page=2]{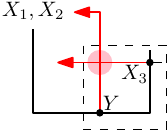}
\caption{Illustration of the proof of Lemma~\ref{lem:Y-inside-Xs}.}
\label{fig:T13proof}
\end{figure}

By Lemma~\ref{lem:X-not-on-boundary} and Lemma~\ref{lem:Y-inside-Xs}, $Y$ is mapped to one of the boxes~$B_{-1}$, $B_0$, or~$B_1$.
By Lemma~\ref{lem:X-distinct-boxes} we may assume that~$X_1,X_2,X_3$ appear in distinct boxes in exactly this order from left to right and also from top to bottom, and none of them is in~$B_{-3}$ or~$B_3$.
Moreover, from Lemma~\ref{lem:Y-inside-Xs} we conclude that~$X_1$ and~$X_3$ are in other boxes than~$Y$, so at most~$Y$ and~$X_2$ are in the same box.
We now distinguish two cases.

\paragraph{Case 1:} $Y$ and~$X_2$ are mapped to the same box.
By symmetry, we may assume that they are mapped to~$B_{1}$ and that~$X_2$ lies above and to the right of~$Y$.
Then the vertex~$X_3$ must be mapped to the box~$B_2$; see Figure~\ref{fig:T13proof3}.
If $YX_3$ were an \necorner{}-edge, then it would cross the edge incident to the bottom of~$X_2$.
It follows that $YX_3$ is an \swcorner{}-edge.
Note that the edge incident to the right of~$X_2$ can only connect to a leaf~$L$ that is mapped to~$B_2\cup B_3$, and~$L$ must be mapped to the right of~$X_3$, as otherwise the edges~$X_2L$ and~$YX_3$ would cross.
The edges that are incident to the bottom and right of~$X_3$ can only connect to points from~$B_2\cup B_3$, so together with~$X_3$ and~$L$ we already have four vertices that are mapped to~$B_2\cup B_3$.
Consequently, the edge incident to the top of~$X_3$ must connect to a point outside of~$B_1\cup B_2\cup B_3$, and therefore this edge crosses the edge~$X_2L$ (see the marked crossing in the figure), a contradiction.

\begin{figure}[htb]
\centering  
\hbox{}
\hfill
\begin{subfigure}[t]{.45\textwidth}
\centering
\includegraphics[page=3]{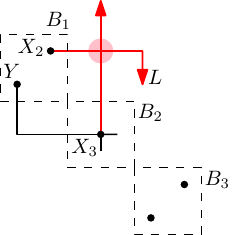}
\caption{Case~1}
\label{fig:T13proof3}  
\end{subfigure}
\hfill
\begin{subfigure}[t]{.45\textwidth}
\centering
\includegraphics[page=4]{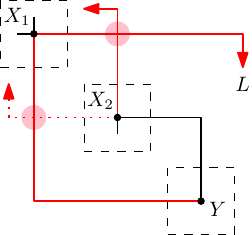}
\caption{Case~2}
\label{fig:T13proof4}  
\end{subfigure}
\hfill
\hbox{}
  
\caption{Illustration of the proof of Theorem~\ref{thm:T13}.}
\label{fig:T13proof34}
\end{figure}

\paragraph{Case 2:} $Y$ and~$X_2$ are mapped to distinct boxes, so all four points~$X_1,X_2,\allowbreak{}X_3,Y$ are in different boxes.
By symmetry, we assume that~$X_1$ and~$X_2$ both lie above and to the left of~$Y$, and~$X_3$ lies below and to the right of~$Y$.
Moreover, we assume that $YX_1$ is an \swcorner{}-edge and that $YX_2$ is an \necorner{}-edge; see Figure~\ref{fig:T13proof4}.
Note that~$X_2$ cannot connect to any points right of~$Y$, and~$X_1$ can only connect to such points by the edge incident to the right of~$X_1$.
As~$Y$ is either mapped to~$B_0$ or~$B_1$, there are at most 7~points above and to the left of~$Y$.
Therefore, as~$X_1$ and~$X_2$ together with their leaves form a set of 8~points, $Y$ must be mapped to~$B_1$, and exactly one leaf~$L$ of~$X_1$ is mapped to a point right of~$Y$, connected to~$X_1$ via an \necorner{}-edge.
Note that~$X_2$ cannot be mapped to~$B_0$, as then the edge incident to the bottom of~$X_2$ could not connect to any point without either crossing~$YX_1$ or~$YX_2$.
Consequently, $X_2$ is mapped to~$B_{-1}$.
However, as~$B_{-1}$ and~$B_0$ together contain only 3~points, and~$X_2$ together with its leaves form a set of 4~vertices, at least one of the two edges incident to the left or top of~$X_2$ must connect to a point above or left of~$X_1$, and this edge will cross either the edge~$YX_1$ or~$X_1L$ (creating one of the two marked crossings in the figure), again a contradiction.
\medskip

In both cases we obtain a contradiction to the assumption that~$T_{13}$ admits an L-shaped embedding in the point set~$S_{13}$.
This completes the proof of Theorem~\ref{thm:T13}.

\section{Proof of Theorem~\ref{thm:T10}}
\label{sec:T10proof}

Consider the ordered tree~$T_{10}$ and the point set~$S_{10}$ depicted in Figure~\ref{fig:T10}.
We label the two degree-4 vertices of~$T_{10}$ by~$X_1,X_2$ and the two degree-2 vertices incident to~$X_1$ and~$X_2$ by~$I_1$ and~$I_2$, respectively.
Moreover, we label the leaves adjacent to~$X_1$ and~$X_2$ by $L_1,L_1',L_2,L_2'$, and the leaves adjacent to~$I_1$ and~$I_2$ by~$L_1''$ and~$L_2''$, as shown in the figure.
We label the points of the point set~$S_{10}$ from left to right by~$P_1,\ldots,P_{10}$.
Note the symmetry of~$T_{10}$, and observe that~$S_{10}$ has reflection symmetries along both diagonals of the grid.

For the sake of contradiction, we assume that an order-preserving L-shaped embedding of~$T_{10}$ in~$S_{10}$ exists.
Clearly, none of the degree-4 vertices~$X_1,X_2$ can be mapped to any of the four points $P_1,P_2,P_9,P_{10}$ which lie on the bounding box of the point set~$S_{10}$.
We also claim that $X_1,X_2$ cannot be mapped to~$P_3$ or~$P_8$.
By symmetry, it suffices to exclude the case that~$X_1$ is mapped to~$P_3$.
In this case, we may assume by symmetry that the edge~$X_1X_2$ is incident to the right of~$X_1$.
Consequently, due to the cyclic order of the neighbors of~$X_1$, $I_1$ must be mapped to~$P_1$ and $L_1$ must be mapped to~$P_2$.
Then $L_1''$ cannot be mapped to any point, a contradiction.

\begin{figure}
\centering
\includegraphics[page=2]{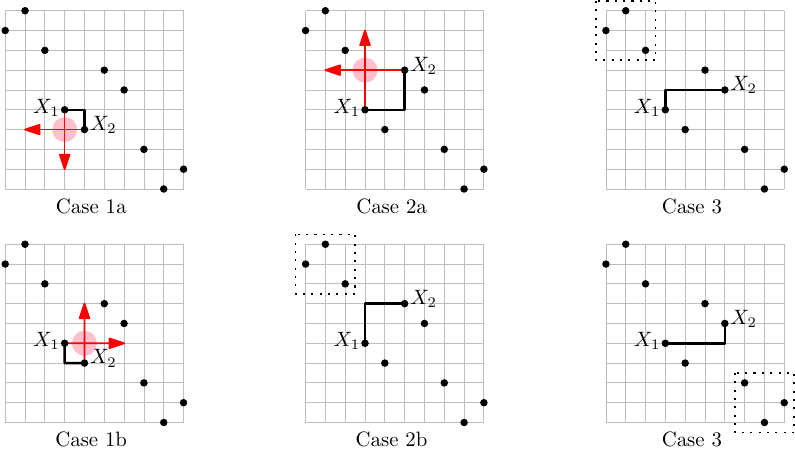}
\caption{Illustration of the proof of Theorem~~\ref{thm:T10}.}
\label{fig:T10proof}
\end{figure}

It follows that~$X_1$ and~$X_2$ are mapped to the points~$P_4,P_5,P_6,P_7$.
By symmetry, we may assume that~$X_1$ is mapped to~$P_4$.
We now distinguish six cases, illustrated in Figure~\ref{fig:T10proof}:
\begin{itemize}[leftmargin=3ex,itemsep=0ex,parsep=0ex,partopsep=0ex,topsep=1mm]
\item Case 1a: $X_2$ is mapped to~$P_5$ and $X_1X_2$ is an \necorner{}-edge.
In this case, the edge incident to the bottom of~$X_1$ and the edge incident to the left of~$X_2$ must cross, a contradiction.

\item Case 1b: $X_2$ is mapped to~$P_5$ and $X_1X_2$ is an \swcorner{}-edge.
In this case, the edge incident to the right of~$X_1$ and the edge incident to the top of~$X_2$ must cross, a contradiction.

\item Case 2a: $X_2$ is mapped to~$P_6$ and $X_1X_2$ is an \secorner{}-edge.
In this case, the edge incident to the top of~$X_1$ and the edge incident to the left of~$X_2$ must cross, a contradiction.

\item Case 2b: $X_2$ is mapped to~$P_6$ and $X_1X_2$ is an \nwcorner{}-edge.
Clearly, none of the four vertices $L_1'$, $L_2'$, $I_1$, and $I_2$ can be mapped to~$P_1$, $P_2$, or~$P_3$.
We claim that $L_1''$ and $L_2''$ cannot be mapped to any of these points either.
By symmetry, it suffices to show the argument for~$L_1''$:
Indeed, $X_1I_1$ is an \swcorner{}-edge, and~$I_1$ can only be mapped to one of~$P_5$, $P_8$, $P_9$, or~$P_{10}$.
If $L_1''$ is mapped to one of~$P_1$, $P_2$, or~$P_3$, then~$I_1$ and~$L_1''$ must be joined via an \necorner{}-edge.
Consequently, if $I_1$ is mapped to~$P_5$, then the edge~$I_1L_1''$ intersects the edge~$X_1X_2$.
On the other hand, if $I_1$ is mapped to~$P_8$, $P_9$ or~$P_{10}$, then together the two edges~$X_1I_1$ and $I_1L_1''$ prevent at least one of the two points~$P_9,P_{10}$ from being reachable from~$X_2$ via one or two L-shaped edges.
Indeed, given the two edges~$X_1I_1$ and $I_1L_1''$, then neither~$P_9$ nor~$P_{10}$ can be reached from~$X_2$ via a single edge, and the only way to reach one of these points via two edges from~$X_2$ is to first take a \necorner{}-edge incident to the top of~$X_2$, but the edge incident to the top of~$X_2$ must lead to the leaf~$L_2$.
This completes the argument that $L_1''$ and $L_2''$ cannot be mapped to~$P_1$, $P_2$, or~$P_3$.
Consequently, only two vertices, namely~$L_1$ and~$L_2$ can be mapped to the three points~$P_1,P_2,P_3$, a contradiction.

\item Case 3: $X_2$ is mapped to~$P_7$.
The subcases where~$X_1X_2$ is an \nwcorner{}-edge or an \secorner{}-edge are symmetric, so it suffices to consider the first one.
In this case we can argue as in Case~2b that only~$L_1$ and~$L_2$ can be mapped to the three points~$P_1,P_2,P_3$, a contradiction.
\end{itemize}

In each case we obtain a contradiction, so this completes the proof of Theorem~\ref{thm:T10}.

\section{Proof of Theorem~\ref{thm:Tr}}
\label{sec:Trproof}

Throughout this section, we assume that $r\geq 10$ is even and $n=9r+8$.
We label the degree-4 vertices of the ordered $n$-vertex tree~$T_r^*$ along the central path by $X_0,\ldots,X_{r+1}$, and for any vertex~$X_i$, $1\leq i\leq r$, we label its two neighbors of degree~4 not on the central path by~$X_i'$ and~$X_i''$, as shown in Figure~\ref{fig:Tr2}.
For our later arguments it will be convenient to orient the edges of~$T_r^*$ which are not on the central path.
Edges incident to a leaf are oriented away from the leaf and edges~$X_i'X_i$ and~$X_i''X_i$ are oriented towards~$X_i$.
In an embedding of the tree, any L-shaped oriented edge appears in one of eight possible orientations, and four of them are important for our proofs; we refer to them as an \lbarc{}-, \ltarc{}-, \rbarc{}-, or \rtarc{}-edge, respectively, where the arrow marks the tip of the oriented edge.

\begin{figure}
\centering
\includegraphics[page=1]{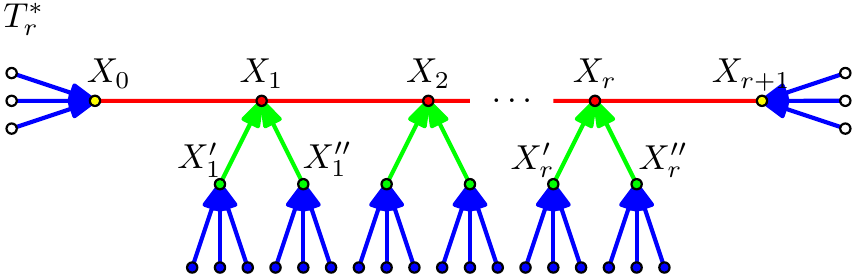}
\caption{Labeling of vertices of the ordered tree~$T_r^*$ for the proof of Theorem~\ref{thm:Tr}.}
\label{fig:Tr2}
\end{figure}

Lemma~\ref{lem:box} immediately gives the following result.

\begin{lemma}
\label{lem:Xp-distinct-boxes}
Each of the degree-4 vertices~$X_i$ for $0\leq i\leq r+1$, and $X_i'$, $X_i''$ for $1\leq i\leq r$, is mapped to a distinct box of the $n$-point $(2,\ldots,2)$-staircase.
\end{lemma}

We refer to the sequence of \necorner{}- or \swcorner{}-edges connecting the central path vertices $X_0,\ldots,X_{r+1}$ as the \emph{spine}.
By symmetry, we may assume w.l.o.g.\ that $X_0$ is mapped to a box on the left of~$X_1$.
In the following we distinguish two main cases, depending on whether $X_0X_1$ is an \necorner{}-edge or an \swcorner{}-edge.

\subsection[Case 1]{Case 1: $X_0X_1$ is an \necorner{}-edge}

Throughout this section, we assume that $X_0X_1$ is an \necorner{}-edge.
Lemma~\ref{lem:Xp-distinct-boxes} and the cyclic order of neighbors around each of the vertices~$X_i$, $i=0,\ldots,r+1$, now enforce a particular shape of all tree edges that connect two degree-4 vertices, as captured by the following lemma; see Figure~\ref{fig:spine-straight}.

\begin{figure}[h]
\centering
\includegraphics[page=2]{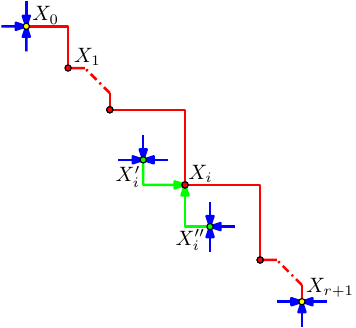}
\caption{Illustration of Lemma~\ref{lem:spine-straight}.}
\label{fig:spine-straight}
\end{figure}

\begin{lemma}
\label{lem:spine-straight}
The vertices $X_0,\ldots,X_{r+1}$ appear exactly in this order from left to right, and any two consecutive such vertices are connected by an \necorner{}-edge.
Moreover, for $i=1,\ldots,r$,
\begin{itemize}[leftmargin=3ex,itemsep=0ex,parsep=0ex,partopsep=0ex,topsep=1mm]
\item the vertices $X_i'$ and~$X_i$ are connected by an \rbarc{}-edge;
\item the vertices $X_i''$ and~$X_i$ are connected by an \lbarc{}-edge;
\item the end segments of the three edges directed from the leaves towards the vertices $X_i'$, $X_i''$, $X_0$, and $X_{r+1}$ form a~\tbot{}, \tright{}, \tleft{}, and \ttop{}, respectively.
\end{itemize}
\end{lemma}

By Lemma~\ref{lem:Xp-distinct-boxes}, each box containing one of the~$X_i$, $1\leq i\leq r$, contains a second point to which a leaf is mapped.
We denote this point by~$P_i$.
Combining Lemmas~\ref{lem:Xp-distinct-boxes} and~\ref{lem:spine-straight} yields the following lemma, which is illustrated in Figure~\ref{fig:blocker}.

\begin{lemma}
\label{lem:blocker}
For every point $P_i$ below the spine exactly one of the following four conditions holds:
\begin{itemize}[leftmargin=3ex,itemsep=0ex,parsep=0ex,partopsep=0ex,topsep=1mm]
\item $P_i$ is connected to~$X_0$ by an \lbarc{}-edge;
\item $P_i$ is connected to~$X_{r+1}$ by an \rbarc{}-edge;
\item there is an index~$j$, $1\leq j < i$, such that $P_i,X_j'',X_j$ are joined by two consecutive \lbarc{}-edges;
\item there is an index~$j$, $i < j\leq r$, such that $P_i,X_j',X_j$ are joined by two consecutive \rbarc{}-edges.
\end{itemize}
For every point~$P_i$ above the spine exactly one of the following two conditions holds:
\begin{itemize}[leftmargin=3ex,itemsep=0ex,parsep=0ex,partopsep=0ex,topsep=1mm]
\item there is an index~$j$, $1\leq j\leq r$, such that $P_i,X_j'$ and $X_j',X_j$ are joined by an \ltarc{}-edge and an \rbarc{}-edge, respectively, wrapping around the top left end of the spine;
\item there is an index~$j$, $1\leq j\leq r$, such that $P_i,X_j''$ and $X_j'',X_j$ are joined by an \rtarc{}-edge and an \lbarc{}-edge, respectively, wrapping around the bottom right end of the spine.
\end{itemize}
\end{lemma}

\begin{figure}
\centering
\includegraphics[page=4]{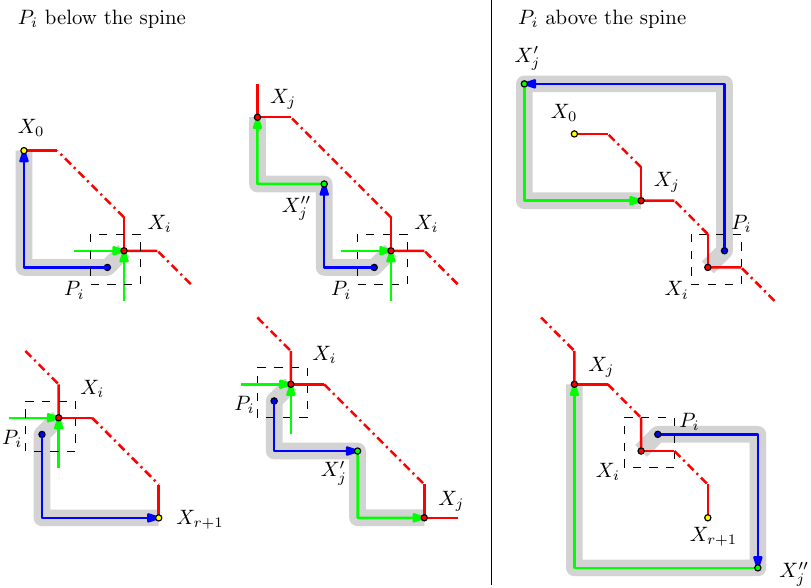}
\caption[]{Illustration of the six different cases in Lemma~\ref{lem:blocker}.
The corresponding blockers are highlighted with bold lines.
}
\label{fig:blocker}
\end{figure}

Consider any pair of points~$P_i,X_0$ as in Lemma~\ref{lem:blocker} connected by an \lbarc{}-edge.
We refer to this edge together with the short diagonal line joining the points~$X_i$ and~$P_i$ in the same box (this line is not part of the tree embedding), as a \emph{\lbarc{}-blocker starting at~$X_i$ and ending at~$X_0$}; see Figure~\ref{fig:blocker}.
Similarly, given any triple of points $P_i,X_j'',X_j$ as in Lemma~\ref{lem:blocker} joined by two consecutive \lbarc{}-edges, we refer to these two edges together with the line joining~$X_i$ and~$P_i$, as a \emph{\lblock{}-blocker{} starting at~$X_i$ and ending at~$X_j$}.
Moreover, given any triple of points $P_i,X_j',X_j$ as in Lemma~\ref{lem:blocker} joined by an \ltarc{}-edge followed by an \rbarc{}-edge wrapping around the top left end of the spine, we refer to these two edges together with the line joining~$X_i$ and~$P_i$, as a \emph{\lwblock{}-blocker{} starting at~$X_i$ and ending at~$X_j$}.
The terms \emph{\rbarc{}-blocker}, \emph{\rblock{}-blocker} and \emph{\rwblock{}-blocker} are defined analogously; see the bottom part of Figure~\ref{fig:blocker}.
The right hand side of Figure~\ref{fig:blocker} shows that for a \lwblock{}-blocker{} or a \rwblock{}-blocker, there is no constraint on~$j$ for a given~$i$ (other than $1\leq j\leq r$).
Observe also that no tree edge can cross a blocker.

For every index $i_1$, $1\leq i_1\leq r$, we define a finite sequence of blockers as follows; see Figure~\ref{fig:seq}:
For $j=1,2,\ldots$ we consider the point~$X_{i_j}$ and the blocker starting at~$X_{i_j}$.
The endpoint~$X_{i_{j+1}}$ of this blocker defines the next index~$i_{j+1}$.
If $i_{j+1}\notin \{i_1,\ldots,i_j\}\cup\{0,r+1\}$, we repeat this process, otherwise we stop.
This yields a finite sequence of indices $i_1,i_2,\ldots,i_\ell$, such that any two consecutive points~$X_{i_j}$ and~$X_{i_{j+1}}$ are joined by a blocker starting at~$X_{i_j}$ and ending at~$X_{i_{j+1}}$.
Clearly, $i_2,\ldots,i_\ell$ all depend on the choice of~$i_1$.
Moreover, by the termination condition above we either have~$i_\ell\in\{i_1,\ldots,i_{\ell-1}\}$ if the blockers close cyclically, or~$i_\ell\in\{0,r+1\}$ if the last blocker ends at~$X_0$ or~$X_{r+1}$ (the terminal index is included in the sequence).
These two cases are illustrated in Figure~\ref{fig:seq}.
We refer to the sequence of blockers generated in this fashion as the \emph{blocker sequence starting at~$X_{i_1}$}.

\begin{figure}
\centering
\includegraphics[page=5]{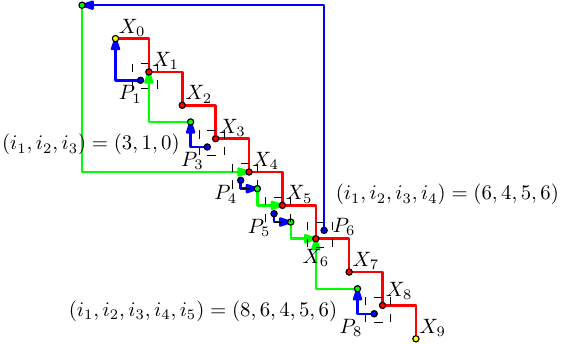}
\caption{Illustration of the definition of blocker sequences.
The figure shows three blocker sequences, one starting at~$X_3$ and ending at~$X_0$ with corresponding index sequence $(i_1,i_2,i_3)=(3,1,0)$, one starting and ending at~$X_6$ with corresponding index sequence $(i_1,i_2,i_3,i_4)=(6,4,5,6)$, and one starting at~$X_8$ and ending at~$X_6$ with corresponding index sequence $(i_1,i_2,i_3,i_4,i_5)=(8,6,4,5,6)$.
}
\label{fig:seq}
\end{figure}

The statement and proof of the following key lemma are illustrated in Figure~\ref{fig:wrap}.

\begin{lemma}
\label{lem:wrap}
Let $1\leq a<b\leq r$ be such that~$P_a$ and~$P_b$ are two consecutive points above the spine each contained in a \rwblock{}-blocker, and let $X_k$ and $X_\ell$ be the blocker endpoints, respectively.
Then there are indices~$c,d$ with $k<c<d\leq \ell$ such that~$P_c$ and~$P_d$ are above the spine. \\
Symmetrically, if~$P_a$ and $P_b$, $1\leq a<b\leq r$, are two consecutive points above the spine each contained in a \lwblock{}-blocker, then there are indices~$c,d$ with $k\leq c<d<\ell$ such that~$P_c$ and~$P_d$ are above the spine.
\end{lemma}

Observe that this lemma does not make any assertions about the relative positions of the points in~$\{X_a,X_b\}$ and~$\{X_k,X_\ell\}$.
In particular, it does not make any assertions about the disjointness of the sets~$\{P_a,P_b\}$ and~$\{P_c,P_d\}$.

\begin{figure}
\centering
\includegraphics[page=6]{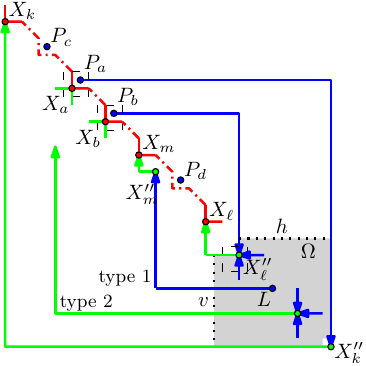}
\caption{Illustration of Lemma~\ref{lem:wrap}.}
\label{fig:wrap}
\end{figure}

\begin{proof}
It suffices to prove the first part of the lemma where $P_a$ and $P_b$ are both contained in a \rwblock{}-blocker.
The second part follows by symmetry.
Let~$h$ denote the horizontal line segment slightly above the box containing~$X_\ell''$ between the two vertical segments of the \rtarc{}-edges leaving~$P_a$ and~$P_b$.
Let~$v$ denote the vertical line segment slightly left of the box containing~$X_\ell''$ between the two horizontal segments of the \lbarc{}-edges leaving~$X_k''$ and~$X_\ell''$.
Let~$\Omega$ denote the region enclosed by the two \rwblock{}-blockers starting at~$X_a$ and~$X_b$ and between the segments~$h$ and~$v$, without the point~$X_k''$.
Note that~$\Omega$ contains~$X_\ell''$ and also the second point in its box, but neither~$X_k''$ nor the second point in its box, so $\Omega$ contains an even number of points from the $(2,\ldots,2)$-staircase.
Observe also that no edge crosses the segment~$h$, as~$P_a$ and~$P_b$ are consecutive points above the spine.
Consider an edge crossing the segment~$v$.
By Lemma~\ref{lem:spine-straight}, this can only be an \lbarc{}-edge starting at a leaf~$L$ in~$\Omega$ and ending at a vertex~$X_m''$ for some~$m$, $k<m<\ell$ (type~1), or an \lbarc{}-edge starting at some~$X_m''$ in~$\Omega$, $k<m<\ell$, and ending at~$X_m$ (type~2).
Figure~\ref{fig:wrap} gives an illustration of both types of edges.
In the case of a type~2 edge, all three leaves adjacent to~$X_m''$ must also be in~$\Omega$.
Therefore, every type~1 edge contributes~1 to the number of vertices in~$\Omega$, and every type~2 edge contributes~4 to the number of vertices in~$\Omega$.
Note that the other two leaves adjacent to~$X_\ell''$ apart from~$P_b$ must also be in~$\Omega$, so~$X_\ell''$ together with these two leaves contributes~3 to the number of vertices in~$\Omega$.
As the number of points from the $(2,\ldots,2)$-staircase in~$\Omega$ is even, there must be at least one type~1 edge starting at a leaf~$L$ in~$\Omega$ and ending at a vertex~$X_m''$, $k<m<\ell$.

By Lemma~\ref{lem:spine-straight}, $X_m''$ is connected to~$X_m$ by another \lbarc{}-edge.
Now consider the blocker sequence starting at~$X_m$.
We prove that it must contain a \lwblock{}-or \rwblock{}-blocker.
For the sake of contradiction suppose not.
Then it can only have \lblock{}-blockers, but no \lbarc{}-, \rbarc{}-, or \rblock{}-blockers:
Indeed, an \lbarc{}-blocker would lead to~$X_0$, which is impossible because of the \lbarc{}-edge between~$X_k''$ and~$X_k$ that shields this blocker sequence from the left.
Moreover, an \rbarc{}- or \rblock{}-blocker would force one of the points $X_i$, $1\leq i\leq r+1$, to lie inside $\Omega$, which is impossible.
However, if the blocker sequence consists only of \lblock{}-blockers, then it must end at~$X_0$, which is again impossible.
This proves our claim that the blocker sequence starting at~$X_m$ contains a \lwblock{}- or \rwblock{}-blocker, and the first such blocker in the sequence will contain the desired point~$P_c$, $k<c\leq m$ (if the very first blocker is of this type then~$c=m$).

An analogous argument applies to the blocker sequence starting at~$X_\ell$.
As the \lbarc{}-edge between~$L$ and~$X_m''$ shields this blocker sequence from the left, the first \lwblock{}- or \rwblock{}-blocker in this sequence contains the desired point~$P_d$, $m<d\leq \ell$.
This completes the proof of the lemma. 
\end{proof}

We will later use the following corollary of Lemma~\ref{lem:wrap}.

\begin{corollary}
\label{cor:wrap}
Suppose there are in total $\alpha\geq 2$ points $P_{i_1},\ldots,P_{i_\alpha}$, $i_1<\cdots<i_\ell$, above the spine each contained in a \rwblock{}-blocker, and let $X_k$ be the endpoint of the blocker starting at~$X_{i_1}$.
Then we have~$k<i_1$, and there are at least $2(\alpha-1)$ many points~$P_i$ with~$i>k$ above the spine. \\
Symmetrically, suppose there are in total $\alpha\geq 2$ points $P_{i_1},\ldots,P_{i_\alpha}$, $i_1<\cdots<i_\ell$, above the spine each contained in a \lwblock{}-blocker, and let $X_k$ be the endpoint of the blocker starting at~$X_{i_\alpha}$.
Then we have~$k>i_\alpha$, and there are at least $2(\alpha-1)$ many points~$P_i$ with~$i<k$ above the spine.
\end{corollary}

\begin{proof}
By symmetry, it suffices to prove the statement for \rwblock{}-blockers.
For $j=1,\ldots,\alpha$, let $X_{k_j}$ be the endpoint of the blocker starting at~$X_{i_j}$.
As these blockers do not intersect, we have $k_1<k_2<\cdots<k_\alpha$ and consequently the intervals $]k_j,k_{j+1}]$, $j=1,\ldots,\alpha-1$, are pairwise disjoint.
Applying Lemma~\ref{lem:wrap} to the pair of consecutive points $P_{i_j},P_{i_{j+1}}$, $j=1,\ldots,\alpha-1$, shows that there are at least two points~$P_i$ with $i\in\left]k_j,k_{j+1}\right]$ above the spine.
Overall, this gives $2(\alpha-1)$ points~$P_i$ with $i>k_1=k$ above the spine.
As $2(\alpha-1)\geq \alpha$, we must have $k<i_1$.
\end{proof}

Consider the collection of all blocker sequences starting at any of the points~$X_i$, $1\leq i\leq r$.
Any blocker in one of these sequences encloses a region together with the spine, and if this region touches a spine edge from the bottom left, then we say that this spine edge is \emph{enclosed}.
Any spine edge that is not enclosed is called \emph{free}.
In Figure~\ref{fig:free}, enclosed regions are shaded.
For any point~$A$ of the staircase point set, consider the second point~$A'$ in the same box of the staircase, and let $L(A)$ denote the halfplane containing those two points, such that the points lie slightly to the left of the boundary of the halfplane.
We define the halfplane~$R(A)$ analogously, by changing left and right in the previous definition.

\begin{lemma}
\label{lem:free}
There is no valid embedding of~$T_r^*$ with zero or one free spine edges.
\end{lemma}

\begin{proof}
We choose two particular degree-4 vertices~$A$ and~$B$ of $T_r^*$ as follows:
If there are no \lwblock{}-blockers, then $A:=X_0$, and otherwise $A$ is defined as the middle vertex of the outermost \lwblock{}-blocker.
Similarly, if there are no \rwblock{}-blockers, then $B:=X_{r+1}$, and otherwise $B$ is defined as the middle vertex of the outermost \rwblock{}-blocker; see Figure~\ref{fig:free}.
Note that if $A=X_0$ and the spine edge~$X_0X_1$ is enclosed, then the edge incident to the bottom of~$X_0$ is part of a \lbarc{}-blocker.
Similarly, if $B=X_{r+1}$ and the spine edge~$X_rX_{r+1}$ is enclosed, then the edge incident to the left of~$X_{r+1}$ is part of a \rbarc{}-blocker.

We first assume that there is no free spine edge.
Consider the regions~$L(A)$ and~$R(B)$.
Note that~$A$ and exactly two of the leaves adjacent to it lie in~$L(A)$, and that $B$ and exactly two of the leaves adjacent to it lie in~$R(B)$.
On the other hand, both regions contain an even number of points from the $(2,\ldots,2)$-staircase.
This immediately yields a contradiction, as none of the vertices $X_i',X_i''$, $1\leq i\leq r$, or any of the leaves adjacent to them can reach into $L(A)$ or~$R(B)$; see the left hand side of Figure~\ref{fig:free}.

It remains to consider the case that there is one free spine edge $X_cX_{c+1}$, $0\leq c\leq r$.
In the following we only consider the subcase $1\leq c\leq r-1$; see the right hand side of Figure~\ref{fig:free}.
The remaining subcases $c=0$ and $c=r$ are symmetric, and can be handled analogously.
We again consider the regions~$L(A)$ and~$R(B)$.
As $X_cX_{c+1}$ is the only free spine edge, at least one of the vertices~$X_c''$, $X_{c+1}'$ or one of the leaves adjacent to one of them must be inside~$L(A)$, and one of them must be inside~$R(B)$.
By symmetry, we may assume that the \lbarc{}-edge from~$X_c''$ to~$X_c$ or the \lbarc{}-edge entering~$X_c''$ has its starting point in~$R(B)$.
This prevents the \rbarc{}-edge from~$X_{c+1}'$ to~$X_{c+1}$ and the \rbarc{}-edge entering~$X_{c+1}'$ from reaching into~$L(A)$.
In this situation the starting point of the \rtarc{}-edge entering~$X_c''$ is the only one that can reach into~$L(A)$, wrapping around the entire spine, which forces $X_c''$ to be in~$R(B)$.
This, however, leads to a contradiction, as only~3 vertices would be mapped to points in~$R(X_c'')$.
\end{proof}

\begin{figure}
\centering
\includegraphics[page=7]{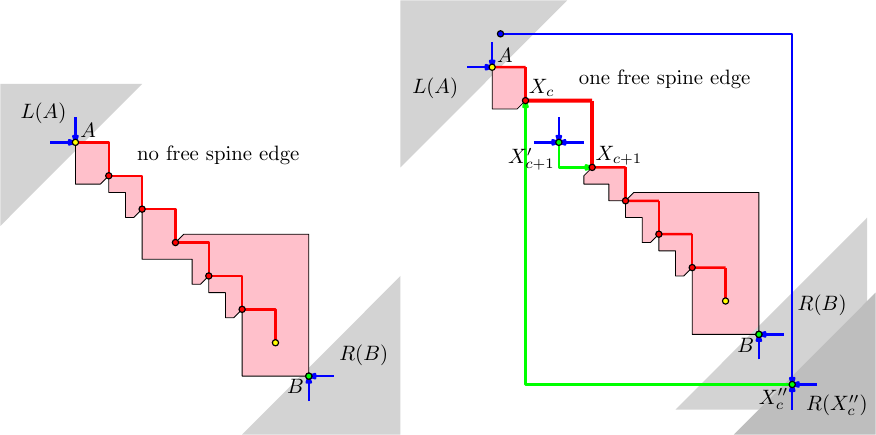}
\caption{Illustration of Lemma~\ref{lem:free}.
Regions enclosed by blockers and the spine are shaded.
}
\label{fig:free}
\end{figure}

With Corollary~\ref{cor:wrap} and Lemma~\ref{lem:free} in hand, we are ready to complete the proof of Theorem~\ref{thm:Tr} in the case that $X_0X_1$ is an \necorner{}-edge.
We let $\alpha_L$ and $\alpha_R$ denote the number of points~$P_i$ above the spine that are contained in a \lwblock{}- or \rwblock{}-blocker, respectively, and we define $\alpha:=\alpha_L+\alpha_R$.
Observe that by the second part of Lemma~\ref{lem:blocker}, for every point~$P_i$ above the spine, the corresponding point~$X_i$ in the same box is the starting point of a \lwblock{}- or \rwblock{}-blocker, so $\alpha$ is the total number of points~$P_i$ above the spine.
Moreover, when considering the points~$P_i$ above the spine from left to right, then we first encounter all those that are contained in a \lwblock{}-blocker, and then all those that are contained in a \rwblock{}-blocker.
By symmetry we may assume that $\alpha_L\leq \alpha_R$.
In the following we distinguish the five cases $\alpha\in\{0,1,2,3,4\}$ and the case $\alpha>4$, and we show that none of them can occur.

\paragraph{Case $\alpha=0$:}

We claim that in this case, exactly one of the spine edges $X_cX_{c+1}$, $0\leq c\leq r$, is free.
Applying Lemma~\ref{lem:free} will therefore conclude the proof.

Consider the blocker sequences starting at~$X_i$ for all $i=1,\ldots,r$.
Each such blocker sequence contains only \lblock{}-, \rblock{}, \lbarc{}-, and \rbarc{}-blockers, but no \lwblock{}- or \rwblock{}-blockers, and hence it either ends at~$X_0$ or~$X_{r+1}$.
If all blocker sequences end at~$X_{r+1}$, then the blocker sequence starting at~$X_1$ only consists of \rblock{}- and \rbarc{}-blockers, and it encloses all spine edges $X_iX_{i+1}$, $1\leq i\leq r$, i.e., $X_0X_1$ is the only free spine edge and the claim is proved.
Symmetrically, if all blocker sequences end at~$X_0$, then the blocker sequence starting at~$X_r$ only consists of \lblock{}- and \lbarc{}-blockers, and it encloses all spine edges $X_iX_{i+1}$, $0\leq i<r$, i.e., $X_rX_{r+1}$ is the only free spine edge and the claim is proved.
It remains to consider the case that at least one blocker sequence ends at~$X_0$ and at least one ends at~$X_{r+1}$.
We let $c$ be the largest index~$1\leq i\leq r$ for which the blocker sequence starting at~$X_i$ ends at~$X_0$.
By this definition, the blocker sequence starting at~$X_c$ contains no points~$X_i$ for $i>c$, otherwise the blocker sequence starting at such a point~$X_i$ would also end at~$X_0$.
Consequently, the blocker sequence starting at~$X_c$ only consists of \lblock{}- and \lbarc{}-blockers, and it encloses all spine edges $X_iX_{i+1}$, $0\leq i<c$, which entails that all blocker sequences starting at~$X_i$ for any $1\leq i\leq c$ end at~$X_0$ as well.
Consequently, the blocker sequence starting at~$X_{c+1}$, which ends at~$X_{r+1}$ by definition, only consists of \rblock{}- and \rbarc{}-blockers, and it encloses all spine edge $X_iX_{i+1}$, $c+1\leq i\leq r$, which entails that all blocker sequences starting at~$X_i$ for any $c+1\leq i\leq r$ end at~$X_{r+1}$ as well.
We conclude that $X_cX_{c+1}$ is the only free spine edge.

\paragraph{Case $\alpha=1$:}

We only need to consider the case $(\alpha_L,\alpha_R)=(0,1)$.
Let~$P_a$ be the unique point above the spine, i.e., $P_a$ is contained in a \rwblock{}-blocker.
Consider the blocker sequence starting at~$X_a$.
If it ends at~$X_0$, then it encloses all spine edges $X_iX_{i+1}$, $0\leq i\leq r$, i.e., there are no free spine edges, and so we are done with the help of Lemma~\ref{lem:free}.
Otherwise this blocker sequence ends at~$X_a$, enclosing the spine edges $X_iX_{i+1}$, $a\leq i\leq r$.
If $a=1$, then $X_{a-1}X_a=X_0X_1$ is the only free spine edge, and we are done using Lemma~\ref{lem:free}.
Otherwise consider the blocker sequences starting at~$X_i$ for all $1\leq i<a$, which must end either at~$X_0$ or~$X_a$.
We let $c$ be the largest index~$1\leq i<a$ for which the blocker sequence starting at~$X_i$ ends at~$X_0$.
Similarly to the case $\alpha=0$, we obtain that the blocker sequence starting at~$X_c$ encloses all spine edges $X_iX_{i+1}$, $0\leq i<c$, that the blocker sequence starting at~$X_{c+1}$ encloses all spine edges $X_iX_{i+1}$, $c+1\leq i<a$, and that $X_cX_{c+1}$ is the only free spine edge.
Consequently, we are done with the help of Lemma~\ref{lem:free}.

\paragraph{Case $\alpha=2$:}

We only need to consider the cases $(\alpha_L,\alpha_R)=(1,1)$ and $(\alpha_L,\alpha_R)=(0,2)$.
Let~$P_a,P_b$, $a<b$, be the two points above the spine.

We first consider the case $(\alpha_L,\alpha_R)=(1,1)$, i.e., $P_a$ is contained in a \lwblock{}-blocker and~$P_b$ is contained in a \rwblock{}-blocker.
Let $S_a$ and~$S_b$ be the blocker sequences starting at~$X_a$ and~$X_b$, respectively.
Observe that either $S_a$ and~$S_b$ both end at~$X_a$, or both end at~$X_b$, or~$S_a$ ends at~$X_a$ and~$S_b$ ends at~$X_b$.
In the first two cases, there are no free spine edges, so applying Lemma~\ref{lem:free} concludes the proof.
If $b-a=1$, then $X_aX_{a+1}=X_{b-1}X_b$ is the only free spine edge, and we are done with Lemma~\ref{lem:free}.
Otherwise the blocker sequences starting at~$X_i$ for all $a<i<b$ either end at~$X_a$ or~$X_b$.
We let $c$ be the largest index~$a<i<b$ for which the blocker sequence starting at~$X_i$ ends at~$X_a$.
Similarly to the case $\alpha=0$, we obtain that the blocker sequence starting at~$X_c$ encloses all spine edges $X_iX_{i+1}$, $a\leq i<c$, that the blocker sequences starting at~$X_{c+1}$ encloses all the spine edges $X_iX_{i+1}$, $c+1\leq i<b$, and that $X_cX_{c+1}$ is the only free spine edge, so applying Lemma~\ref{lem:free} concludes the proof.

We now consider the case $(\alpha_L,\alpha_R)=(0,2)$, i.e., $P_a$ and~$P_b$ are both contained in a \rwblock{}-blocker.
Let~$X_k$ be the endpoint of the blocker starting at~$X_a$.
From Corollary~\ref{cor:wrap}, we obtain that~$k<a$, i.e., the blocker sequence starting at~$X_a$ must end at~$X_0$, enclosing all spine edges~$X_iX_{i+1}$, $0\leq i\leq r$.
Applying Lemma~\ref{lem:free} again completes the proof.

\paragraph{Case $\alpha=3$:}

We only need to consider the cases $(\alpha_L,\alpha_R)=(1,2)$ and $(\alpha_L,\alpha_R)=(0,3)$.
Let~$P_a,P_b,P_c$, $a<b<c$, be the three points above the spine.

We first consider the case $(\alpha_L,\alpha_R)=(1,2)$, i.e., $P_b,P_c$ are both contained in a \rwblock{}-blocker.
Let~$X_k$ be the endpoint of the blocker starting at~$X_b$.
From Corollary~\ref{cor:wrap}, we obtain that~$k<b$, i.e., the blocker sequence starting at~$X_b$ must end at~$X_a$, together with the blocker sequence starting at~$X_a$, and both enclose all spine edges~$X_iX_{i+1}$, $0\leq i\leq r$.
Consequently, we are done with the help of Lemma~\ref{lem:free}.

We now consider the case $(\alpha_L,\alpha_R)=(0,3)$, i.e., all three points~$P_a,P_b,P_c$ are contained in a \rwblock{}-blocker.
Corollary~\ref{cor:wrap} implies that there are at least $2(\alpha_R-1)=4$ points~$P_i$ above the spine, a contradiction.

\paragraph{Case $\alpha=4$:}

We only need to consider the cases $(\alpha_L,\alpha_R)=(2,2)$, $(\alpha_L,\alpha_R)=(1,3)$, and $(\alpha_L,\alpha_R)=(0,4)$.
Let~$P_a,P_b,P_c,P_d$, $a<b<c<d$, be the four points above the spine.

If $(\alpha_L,\alpha_R)=(2,2)$, then Corollary~\ref{cor:wrap} shows that the blocker sequences starting at~$X_b$ and~$X_c$ cannot coexist:
Specifically, the blocker sequence starting at~$X_b$ with a \lwblock{}-blocker ends at~$X_i$ with $b<i\leq r+1$, and the blocker sequence starting at~$X_c$ with a \rwblock{}-blocker ends at $X_j$ with $0\leq j<c$.
This is a contradiction.
Similarly, if $(\alpha_L,\alpha_R)=(1,3)$, then Corollary~\ref{cor:wrap} shows that the blocker sequences starting at~$X_a$ and~$X_b$ cannot coexist.
If $(\alpha_L,\alpha_R)=(0,4)$, then Corollary~\ref{cor:wrap} implies that there are at least $2(\alpha_R-1)=6$ points~$P_i$ above the spine, a contradiction.

\paragraph{Case $\alpha>4$:}

Corollary~\ref{cor:wrap} shows that there are at least $2(\alpha_L-1)+2(\alpha_R-1)=2\alpha-4$ points~$P_i$ above the spine, which is a contradiction, as $2\alpha-4>\alpha$ for $\alpha>4$.

\subsection[Case 2]{Case 2: $X_0X_1$ is an \swcorner{}-edge}

Throughout this section, we assume that $X_0X_1$ is an \swcorner{}-edge.
Lemma~\ref{lem:Xp-distinct-boxes} and the cyclic order of neighbors around each of the vertices~$X_i$, $i=0,\ldots,r+1$, now enforce a particular shape of all tree edges that connect two degree-4 vertices, as captured by the following lemma; see Figure~\ref{fig:spine-wrap}.

\begin{figure}[h]
\centering
\includegraphics[page=3]{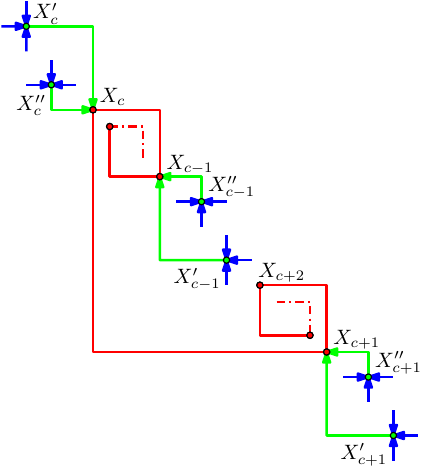}
\caption{Illustration of Lemma~\ref{lem:spine-wrap}.}
\label{fig:spine-wrap}
\end{figure}

\begin{lemma}
\label{lem:spine-wrap}
For $i=0,\ldots,r$, the vertex~$X_i$ is left of~$X_{i+1}$ and both are connected by an \swcorner{}-edge if~$i$ is even, and the vertex~$X_i$ is right of $X_{i+1}$ and both are connected by an \necorner{}-edge if~$i$ is odd.
Moreover, for $i=1,\ldots,r$,
\begin{itemize}[leftmargin=3ex,itemsep=0ex,parsep=0ex,partopsep=0ex,topsep=1mm]
\item the vertices $X_i'$ and~$X_i$ are connected by an \rtarc{}-edge if~$i$ is even, and they are connected by an \lbarc{}-edge if~$i$ is odd;
\item the vertices $X_i''$ and~$X_i$ are connected by an \rbarc{}-edge if~$i$ is even, and they are connected by an \ltarc{}-edge if~$i$ is odd;
\item the end segments of the three edges directed from the leaves towards the vertex~$X_0$ form a \tbot{}, and the end segments of the three edges directed from the leaves towards~$X_i'$ and~$X_i''$ form a~\tleft{} or \tbot{}, respectively, if $i$ is even, and a \tright{} or \ttop{} if~$i$ is odd.
\end{itemize}
\end{lemma}

We define the \emph{length} of a spine edge $X_iX_{i+1}$, $0\leq i\leq r$, in the embedding as the number of boxes in the $(2,\ldots,2)$-staircase between its endpoints plus~1.
For instance, if it connects two neighboring boxes, then its length is~1.
By Lemma~\ref{lem:spine-wrap}, there is a unique longest spine edge $X_cX_{c+1}$, $0\leq c\leq r$, and the two length sequences of the edges~$X_i X_{i+1}$ for $i=c,c+1,\ldots,r$ and $X_{i+1}X_i$ for $i=c,c-1,\ldots,0$ are strictly decreasing, i.e., each of the two corresponding parts of the spine spirals into itself in counterclockwise or clockwise direction, respectively, as shown in Figure~\ref{fig:spine-wrap}.
By the requirement that $r\geq 10$, the longer of these two sequences consists of at least~6 spine edges, and by symmetry we may assume that it is the latter one, i.e., the initial part of the spine looks as shown in Figure~\ref{fig:spine-wrap2}.

\begin{figure}[h]
\centering
\includegraphics[page=8]{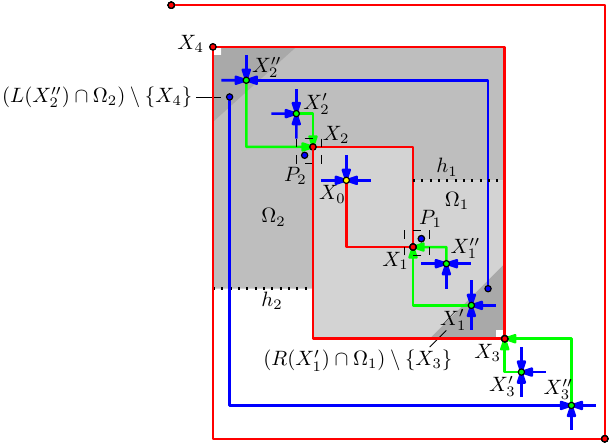}
\caption{Illustration of the proof of Theorem~\ref{thm:Tr} in Case~2.}
\label{fig:spine-wrap2}
\end{figure}

For $i\in\{1,2\}$, we let $h_i$ denote a horizontal line segment between the vertical segments of the spine edges~$X_iX_{i+1}$ and~$X_{i+2}X_{i+3}$, passing above the box containing~$X_1$ and~$P_1$ if $i=1$ and below the box containing~$X_2$ and~$P_2$ if $i=2$, and let $\Omega_i$ denote the region enclosed by the spine and this segment; see the figure.
One of the leaves of the tree~$T_r^*$ must be mapped to the point~$P_1$, which lies in the same box as~$X_1$.
This can only be the leaf adjacent to~$X_2''$ via an \ltarc{}-edge, or the leaf adjacent to~$X_1'$ via an \rtarc{}-edge (if $P_1$ is above and to the right of~$X_1$), or the leaf adjacent to~$X_1''$ via an \rbarc{}-edge (if $P_1$ is below and to the left of~$X_1$).
In the last two cases, the leaf adjacent to~$X_2''$ via an \rtarc{}-edge must lie in the region $(R(X_1')\cap \Omega_1)\setminus\{X_3\}$ or in the region $(R(X_1'')\cap \Omega_1)\setminus\{X_3\}$, respectively.
This is because this region contains an even number of points, and therefore an even number of tree vertices must be mapped to them.
In any case, the edge incident to the right of~$X_2''$ must reach into~$\Omega_1$.
Consequently, the leaf adjacent to~$X_3''$ via an \rbarc{}-edge must lie in the region $(L(X_2'')\cap \Omega_2)\setminus\{X_4\}$, in order to map an even number of tree vertices to this region.
However, as none of the leaves adjacent to~$X_2'$ can connect to~$P_2$, which lies in the same box as~$X_2$, no vertex is mapped to~$P_2$, a contradiction.

This completes the proof of Theorem~\ref{thm:Tr}.

\section{Computer-based proofs of Theorems~\ref{thm:small-trees} and \ref{thm:small-ordered-trees}}
\label{sec:comp}

We implemented a C++ program to test whether a given (unordered or ordered) tree admits an L-shaped embedding in a given set of points.
Our algorithm recursively embeds vertices and edges in all possible ways until either a crossing occurs or a valid drawing is obtained.

Each point set is represented by a permutation, which captures the $y$-coordinates of the points from left to right.
Those permutations are generated in lexicographic order using the C++ standard library function \texttt{next\_permutation}.
When embedding unordered trees, we only need to test point sets that are non-isomorphic up to rotation and mirroring, as an unordered tree is embeddable in a point set if and only if it is embeddable on the rotated or mirrored point set.
This filtering of point sets is achieved by considering only the lexicographically smallest permutation under these two operations.
Similarly, when embedding ordered trees, we may omit testing point sets that are isomorphic up to rotation (but not mirroring).

The list of all non-isomorphic unordered and ordered trees was generated with SageMath~\cite{sagemath_website}, using the integrated nauty graph generator~\cite{McKayPiperno2014}, and then loaded by the C++ program.

When testing ordered trees, we only need to test trees that admit more than one way to cyclically order the neighbors of all vertices, as otherwise the tree is equivalent to the corresponding unordered tree.
Here we consider two ordered trees the same if they differ only in changing the orientation of all cyclic orders from clockwise to counterclockwise or vice versa, which corresponds to mirroring the embedding.

As pairs of trees and point sets can be tested independently, we parallelized our computations; see Table~\ref{tab:run}.
The source code of all those programs is available online~\cite{www}.

\begin{table}[htb]
\caption{
Number of non-isomorphic point sets and unordered/ordered trees with maximum degree~4 up to $n\leq 12$, and the computation times of our C++ program.
The times marked with * are the sum of parallelized computations on 16 cores.
}
\def\arraystretch{1.2}
\centering
\begin{tabular}{r|rrl|rrl}
$n$ & point sets    & unordered trees & CPU time & point sets   & ordered trees & CPU time \\
& \href{http://oeis.org/A903}{OEIS/A903} & \href{http://oeis.org/A602}{OEIS/A602} &  & \href{http://oeis.org/A263685}{OEIS/A263685} & &  \\ 
\hline
4 & 7 & 2 & & 9 & 2 & \\
5 & 23 & 3 & & 33 & 3 & \\
6 & 115 & 5 & & 192 & 5 & \\
7 & 694 & 9 & & 1.272 & 10 & \\
8 & 5.282 & 18 &$< 1$ sec & 10.182 & 21 &$< 1$ sec \\
9 & 46.066 & 35 & 9 sec & 90.822 & 48 & 21 sec \\
10 & 456.454 & 75 & 7 min & 908.160 & 120 & 21 min \\
11 & 4.999.004 & 159 & 12 hours & 9.980.160 & 312 & 64 hours* \\
12 & 59.916.028 & 355 & 84 days* & 119.761.980 & 864 & --- \\
\end{tabular}
\label{tab:run}
\end{table}

\section{Further non-embeddable examples}
\label{sec:examples}

In this section, we present further pairs of (unordered) $n$-vertex trees and sets of $n$~points for $n=13,\allowbreak{}14,\allowbreak{}16,\allowbreak{}17,\allowbreak{}18,19,20$, which do not admit an L-shaped embedding.
The trees $T_n$ are obtained as subtrees of the tree shown in Figure~\ref{fig:tree20}, by taking the subgraph induced by all unlabeled vertices and the vertices with labels~$\leq n$.
The corresponding point sets are encoded below in staircase notation.
Note that all those staircase point sets have rotation and reflection symmetry and boxes of size at most~3.
The fact that those instances do not allow an L-shaped embedding was established with computer help via a SAT solver, as described below.

\begin{figure}[htb]
  \centering
    \includegraphics{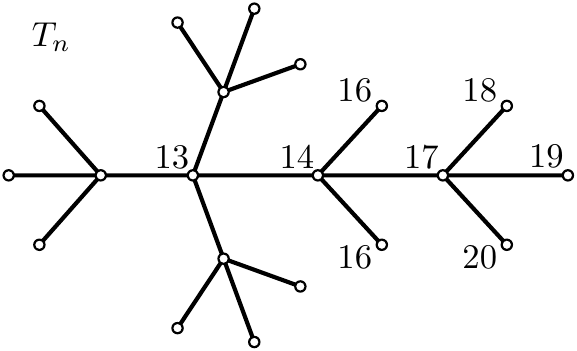}
  \caption{The 20-vertex tree $T_{20}$.}
  \label{fig:tree20}
\end{figure}

\subsection*{$n=13$:}
\begin{verbatim}
  (1,1,2,2,1,2,2,1,1)
  (1,1,3,1,1,1,3,1,1)
  (2,2,2,1,2,2,2)
  (2,3,1,1,1,3,2)
\end{verbatim}

\subsection*{$n=14$:}
\begin{verbatim}
  (1,1,2,1,2,2,1,2,1,1)
  (2,2,1,2,2,1,2,2)
\end{verbatim}

\subsection*{$n=16$:}
\begin{verbatim}
  (1,3,1,1,1,2,1,1,1,3,1)
  (1,3,2,1,2,1,2,3,1)
\end{verbatim}

\subsection*{$n=17$:}
\begin{verbatim}
  (1,1,3,1,1,3,1,1,3,1,1)
\end{verbatim}

\subsection*{$n=18$:}
\begin{verbatim}
  (1,1,2,1,1,1,2,2,1,1,1,2,1,1)
\end{verbatim}

\subsection*{$n=19$:}
\begin{verbatim}
  (1,1,3,1,1,1,3,1,1,1,3,1,1)
  (1,1,3,1,2,3,2,1,3,1,1)
  (1,1,3,2,1,3,1,2,3,1,1)
  (2,3,1,1,1,3,1,1,1,3,2)
  (2,3,1,2,3,2,1,3,2)
  (2,3,2,1,3,1,2,3,2)
\end{verbatim}

\subsection*{$n=20$:}
\begin{verbatim}
  (1,1,2,1,1,1,2,2,2,1,1,1,2,1,1)
  (1,1,2,1,2,2,2,2,2,1,2,1,1)
  (1,1,2,2,1,2,2,2,1,2,2,1,1)
  (2,2,1,1,1,2,2,2,1,1,1,2,2)
  (2,2,1,2,2,2,2,2,1,2,2)
  (2,2,2,1,2,2,2,1,2,2,2)
\end{verbatim}

\subsection{The SAT model}
\label{sec:sat}

To test whether a given tree with vertex set $\{1,\ldots,n\}$ admits an L-shaped embedding in a given point set $\{P_1,\ldots,P_n\}$, we formulated a Boolean satisfiability problem that has a solution if and only if the tree admits an embedding in the point set.

Our SAT model has variables~$x_{i,j}$ to indicate whether the vertex~$i$ is mapped to the point~$P_j$, and for every edge~$ab$ in the tree a variable~$y_{a,b}$ to indicate whether the edge is connected horizontally to~$a$ (otherwise it is connected vertically to~$a$).
The following constraints are necessary and sufficient to guarantee the existence of an L-shaped embedding:

\begin{itemize}[leftmargin=3ex,itemsep=0.4ex,parsep=0ex,partopsep=0ex,topsep=1mm]
\item\label{item:bijection}
\textbf{Injective mapping from vertices to points:}
Each vertex is mapped to a point, and no two vertices are mapped to the same point.
 
\item 
\textbf{L-shaped edges:}
For each edge~$ab$ of the tree, $a$ is either connected horizontally or vertically to~$b$.
Figure~\ref{fig:sat1} gives an illustration.
 
\item
\textbf{No overlapping edge segments:}
For each pair of incident edges $ab$ and~$ac$, if $b$ and $c$ are mapped to the right of~$a$, then $a$ cannot be connected horizontally to both~$b$ and~$c$.
An analogous statement holds if~$b$ and~$c$ are both mapped to the left, above, or below~$a$.
 Figure~\ref{fig:sat2} gives an illustration.
 
\item
\textbf{No crossing edge segments:}
For each pair of edges $ab$ and~$cd$, the vertices $a,b,c,d$ must not be mapped so that segments cross.
More specifically, for each four points $p,q,r,s$ (to which $a,b,c,d$ may map), there are at most four cases that have to be forbidden in the mapping, depending on the relative position of~$p,q,r,s$.
Figures~\ref{fig:sat3} and~\ref{fig:sat4} give an illustration.
\end{itemize}

\begin{figure}[htb]
\centering 

\hbox{}
\hfill
\begin{subfigure}[t]{.22\textwidth}
\centering
\includegraphics[page=1]{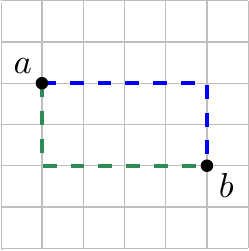}
\caption{}
\label{fig:sat1}  
\end{subfigure}  
\hfill  
\begin{subfigure}[t]{.22\textwidth}
\centering
\includegraphics[page=2]{sat}
\caption{}
\label{fig:sat2}  
\end{subfigure}  
\hfill  
\begin{subfigure}[t]{.22\textwidth}
\centering
\includegraphics[page=3]{sat}
\caption{}
\label{fig:sat3}  
\end{subfigure}
\hfill  
\begin{subfigure}[t]{.22\textwidth}
\centering
\includegraphics[page=4]{sat}
\caption{}
\label{fig:sat4}  
\end{subfigure}  
\hfill  
\hbox{}
  
\caption{
Illustration of the constraints of the SAT model.
}
\label{fig:sat}
\end{figure}

The resulting CNF formula thus has $\Theta(n^2)$ variables and $\Theta(n^4)$ clauses.
Our Python program that creates a SAT instance for a given pair of tree and staircase point set is available online~\cite{www}.
We used the SAT solver PicoSAT~\cite{Biere08}, which allows enumeration of all solutions.
We also made use of pycosat, which provides Python bindings to PicoSAT.

\section{Open problems}
\label{sec:prob}

We currently do not know of any infinite family of (unordered) trees which do not always admit an L-shaped embedding.
However, we conjecture that the instance in Figure~\ref{fig:Tr} is such a family when considering the tree as unordered.
Moreover, since all non-embeddable examples that we know are trees with pathwidth~2 (lobsters), it would be interesting to know whether trees with pathwidth~1 (caterpillars) always admit an L-shaped embedding.
So far all known non-embeddable trees have maximum degree~4, so the question for trees with maximum degree~3 remains open~\cite{KanoSuzuki2011,FinkEtAl2012,GiacomoFFGK13}.

A more general class of embeddings are \emph{orthogeodesic} embeddings, where the edges are drawn with minimal $\ell_1$-length and consist of segments along the grid induced by the point set~\cite{KatzEtAl2010,GiacomoFFGK13,scheucher2015,bbhl16}.
The best known bounds are due to B\'ar\'any et al.~\cite{bbhl16} who showed that every $n$-vertex tree with maximum degree~4 admits an orthogeodesic embedding in every point set of size $\lfloor11n/8\rfloor$.
Unfortunately, the tree~$T_{13}$, which we proved not to admit an L-shaped embedding on the point $S_{13}$, does admit an orthogeodesic embedding on~$S_{13}$ (see Figure~\ref{fig:ortho}), so the question whether $n$~points are always sufficient to guarantee an orthogeodesic embedding of any $n$-vertex tree~\cite{GiacomoFFGK13,bbhl16} also remains open. 

\begin{figure}[htb]
\centering
\includegraphics{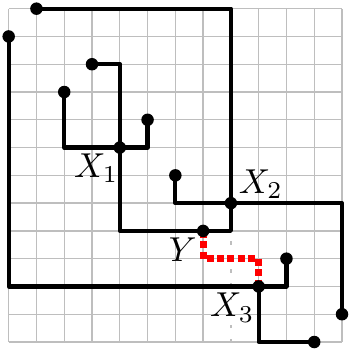}
\caption{An orthogeodesic embedding of the tree~$T_{13}$ in the point set~$S_{13}$.
The only edge with two turns (not L-shaped) is drawn dotted.
}
\label{fig:ortho}
\end{figure}

\section*{Acknowledgements}

Manfred Scheucher was partially supported by DFG Grant FE~340/12-1.
Moreover, we gratefully acknowledge the computing time granted by TBK Automatisierung und Messtechnik GmbH.
Torsten M\"utze was supported by Czech Science Foundation grant GA~19-08554S, and by German Science Foundation grant~413902284.
We thank the anonymous referees for several thoughtful comments that helped improving the presentation of this paper.

\bibliographystyle{alphaabbrv-url}
\bibliography{refs}

\end{document}